%% file: draft2_arxiv_submission.tex
\documentclass[aps,pra,twocolumn,superscriptaddress,longbibliography]{revtex4-2} 
\input{preamble.tex}

\usepackage{amsmath}
\usepackage{enumitem}
\usepackage{tikz}

\begin{document}

\title{Asymptotic scars arising from the remnants of quantum phase transitions}

\author{Leonard Logari\'c \orcidlink{0000-0001-9486-2047}}  
\email{logaricleonard@gmail.com}
\affiliation{International Center for Elementary Particle Physics, The University of Tokyo, 7-3-1 Hongo, Bunkyo-ku, Tokyo 113-0033, Japan} 
\author{Shane Dooley \orcidlink{0000-0002-2856-8840}}
\affiliation{Dublin Institute for Advanced Studies, School of Theoretical Physics,
10 Burlington Road, Dublin, D04 C932, Ireland}
\author{Tsuyoshi Okubo \orcidlink{0000-0003-4334-7293}}
\affiliation{Department of Physics, Niigata University, 8050 Ikarashi 2-no-cho, Nishi-ku, Niigata 950-2181, Japan}
\affiliation{Department of Physics, The University of Tokyo, 7-3-1 Hongo, Bunkyo-ku, Tokyo 113-0033, Japan}
\author{Nobuyuki Yoshioka \orcidlink{0000-0001-6094-8635}}
\email{ny.nobuyoshioka@gmail.com}
\affiliation{International Center for Elementary Particle Physics, The University of Tokyo, 7-3-1 Hongo, Bunkyo-ku, Tokyo 113-0033, Japan}

\date{\today}
 
\begin{abstract}

Many-body localisation, Hilbert space fragmentation and quantum many-body scars (QMBS) are well-established mechanisms of ergodicity breaking in quantum systems. Recently, asymptotic quantum many-body scars (AQMBS), non-thermal states whose stability grows with system size, have been proposed theoretically and probed experimentally. Here, we show that AQMBS arise naturally as remnants of certain quantum phase transitions. For a broad class of frustration-free Hamiltonians, we prove that their gapless excitations survive deformation into a ``scarred’’ Hamiltonian as states whose energy variance is bounded by the excitation gap. Consequently, the gapless excitations at a quantum critical point become AQMBS of the scarred Hamiltonian. Their fidelity relaxation time is then set by the critical scaling of the gap, determined by the dynamic critical exponent. Modelling the critical Fredkin and Motzkin chains and the two-dimensional PVBS model, we obtain AQMBS whose variances directly reflect the dynamic exponents of the underlying transitions. Where the gap of the Motzkin chain closes super-exponentially, we find exponentially stable AQMBS, for which we also observe numerical evidence of volume-law entanglement.

\end{abstract}

\maketitle

\section{Introduction}

The mechanisms for thermalisation of quantum systems, and its absence, have been a focus of intense research efforts. This has resulted in the discovery of several ergodicity breaking mechanisms, including many-body localisation (MBL) \cite{Pon-15, Hor-15, Nan-15a, Aba-19}, Hilbert space fragmentation \cite{Sal-20, Mou-22a, Zha-25a}, as well as quantum many-body scars (QMBS) \cite{Tur-18a, Doo-21a,Mou-22a,Cha-23b,Logaric-24}. While many-body localization (MBL) is expected to render most eigenstates of a quantum system non-thermal, QMBS states represent a vanishingly small fraction of non-thermal eigenstates embedded in an otherwise thermal spectrum. In experimental settings, the presence of these states can lead to long-lived oscillations, slow decay of local observables, as well as anomalously slow entanglement growth \cite{Ber-17,Blu-21,Su-23a,Zha-25a,Zha-23a}.

Recently, asymptotic quantum many-body scars (AQMBS)~\cite{Got-23} have been proposed theoretically as well as probed on a trapped ion device \cite{Log-26}. Unlike conventional QMBS, these states are not eigenstates of the Hamiltonian at any finite system size. Instead, their energy variance  vanishes in the thermodynamic limit, where they become non-thermal from the perspective of the eigenstate thermalisation hypothesis (ETH). As a consequence of the Mandelstam--Tamm bound~\cite{Man-45}, their fidelity relaxation  \emph{diverges} with system size. Although stability that grows with system size may appear pathological, AQMBS have been identified in several models~\cite{Got-23, Ren-24, Kun-25, Log-26}.

A common property of many QMBS is their low entanglement, usually obeying subvolume law scaling. The scars originally found in the PXP ~\cite{Ber-17, Tur-18a}, as well as those identified in other models, typically exhibit subvolume entanglement scaling~\cite{Mou-18, Lin-19, Mou-20a}, in contrast to the volume law of thermal eigenstates. In this respect, QMBS resemble ground states more than typical eigenstates in the middle of the energy spectrum. This analogy was used in Ref.~\cite{Lar-24} to introduce scar phase transitions, by promoting a critical ground state to an QMBS, by suitably modifying the parent Hamiltonian. 

A quantum phase transition, however, is characterised not only by its ground state, but also by the gapless excitations above it. One of the key properties is the rate of the closing of the gap, $\Delta(N) \sim N^{-z}$, with $z$ the dynamic critical exponent. What becomes of these excitations when a critical ground state is promoted to a scar? They do not in general remain eigenstates of the scarred Hamiltonian, however, such gapless excitations have been found to deform into AQMBS~\cite{Kun-25, Log-26}. All numerical and analytical evidence so far has been restricted to quadratically stable AQMBS, with energy variance scaling as $\sim N^{-2}$~\cite{Got-23, Ren-24, Gio-26}, and all known AQMBS obey an area law of entanglement. This leaves open whether the stability of AQMBS can be non-quadratic and whether volume-law AQMBS also exist, in light of recently discovered \emph{volume}-law QMBS \cite{Lan-22, Muk-26, Doo-26a}.


In this work we show that the variance of AQMBS states in deformed, ``scarred'' models can be directly linked to the gap in the ``reference'' Hamiltonian. Exploiting this property, we find both AQMBS with different polynomial stability scalings, as well as exponentially stable AQMBS, arising from the deformed  Motzkin model. For the latter, we also observe numerical evidence for volume law scaling of the bipartite entanglement entropy, similar to the scaling of the ground state entropy. For the investigated models, we observe that the decay rate of the AQMBS is directly dictated by the dynamic exponent of the gapless excitation.

\section{Results} \label{sec:results}

\subsection{Connection between gap and variance}
\label{subsec:general_connection}

We first demonstrate the connection between variance of AQMBS states in a scarred, deformed Hamiltonian and the gap in a reference Hamiltonian. For this purpose, we will consider reference Hamiltonians of the form:
\begin{equation} \label{eq:reference_hamiltonian}
    \hat{\mathbb{H}}_0 = \sum_{n = 0}^{N-1} \hat{P}_n,
\end{equation}
which acts on a set of $N$ qudits. Each projector $\hat{P}_n$ acts on a region $X_n$ containing the qudit labelled $n$ and at most $l$ other contiguous sites. In other words, the reference Hamiltonian is $l+1$ local. We also assume that $\hat{\mathbb{H}}_0$ is frustration-free, i.e., its ground state(-s) $\ket{\mathcal{G}}$ are simultaneously annihilated by each $\hat{P}_n$. We denote by
\begin{equation} \label{eq:z_def_main}
    z = \max_n \bigl| \{ m \neq n : X_m \cap X_n \neq \varnothing \} \bigr|
\end{equation}
the maximal number of projectors whose support overlaps with that of a given projector, and assume that $z$ is independent of $N$. Let us also introduce a set of ``scarred'' Hamiltonians, defined as:
\begin{equation} \label{eq:unc_hamiltonian}
    \hat{\mathbb{H}}_\mathrm{S} = \sum_n \hat{P}_n\hat{h}_n\hat{P}_n = \sum_n \hat{H}_n,
\end{equation}
where $\hat{h}_n$ is an arbitrary Hermitian operator with bounded spectral norm, $||\hat{h}_n || \leq B$, which has a support that is contained within $X_n$. We note that, for the AQMBS to lie in the middle of the energy spectrum, we generically require that $\hat{h}_n$ are not positive semi-definite. Using $\hat{\mathbb{H}}_\mathrm{0}$ and $\hat{\mathbb{H}}_\mathrm{S}$ allows us to state the following theorem, which we prove in Supplementary Material (SM) \ref{app: proof_aqmbs}:
\begin{restatable}{theorem}{aqmbsthm}
\label{theorem:AQMBS_main_text}
Let $\hat{\mathbb{H}}_0$ be a frustration-free, $(l+1)$-local Hamiltonian with eigenstates $\ket{\mathcal{A}_k}$ whose eigenvalues $\Delta_k(N)$ vanish in the thermodynamic limit. Then the states $\ket{\mathcal{A}_k}$ are AQMBS of $\hat{\mathbb{H}}_\mathrm{S} = \sum_n \hat{P}_n \hat{h}_n \hat{P}_n$. Specifically, their energy and variance with respect to $\hat{\mathbb{H}}_\mathrm{S}$ obey:
    \begin{align} \label{eq:main_bounds}
        \bigl| \langle \hat{\mathbb{H}}_\mathrm{S} \rangle \bigr| & \le B \, \Delta_k(N), \\
        \delta^2 \hat{\mathbb{H}}_{\rm S} &\le B^2 \bigl[ \Delta_k^2(N) + 2z \, \Delta_k(N) \bigr],
    \end{align}
    where $\delta^2 \hat{\mathbb{H}}_{\rm S} \equiv \langle \hat{\mathbb{H}}_\mathrm{S}^2 \rangle - \langle \hat{\mathbb{H}}_\mathrm{S} \rangle^2$, $\langle \cdot \rangle \equiv \bra{\mathcal{A}_k} \cdot \ket{\mathcal{A}_k}$, $B$ bounds the spectral norm of $\hat{h}_n$, and $z$ is defined in Eq.~\eqref{eq:z_def_main}. Both quantities therefore vanish in the thermodynamic limit, as $\Delta_k(N) \to 0$.
\end{restatable}
\noindent We also note that for sufficiently small gaps, $\Delta_k \ll z$, this bound reduces to $\delta^2 \mathbb{H}_\mathrm{S} \lesssim B^2 2z \Delta_k (N)$. This shows that, depending on the rate of the gap closing, AQMBS can vary in how rapidly they are becoming eigenstates of the system. Since the gaps of different models are typically monotonically decreasing with system size, this also implies that we progressively smaller variances at finite $N$. By the Mandelstamm-Tamm bound, the fidelity relaxation time is directly related to the state's variance ($\delta ^2 \mathbb{H}$) with respect to the Hamiltonian  \cite{Man-45}:
\begin{equation} \label{eq:tamm_bound}
    T_\mathrm{rel}  \geq \frac{h}{4 \sqrt{\delta^2 \mathbb{H}}}.
\end{equation}
Therefore, the gap in the reference Hamiltonian will dictate the stability of AQMBS states in the scarred Hamiltonian. In order to demonstrate this explicitly, we will first consider the Fredkin spin chain.

\subsection{Fredkin Spin Chain}
\label{subsec:fredkin}

\begin{figure*}
    \centering
    \includegraphics[width=0.98\linewidth]{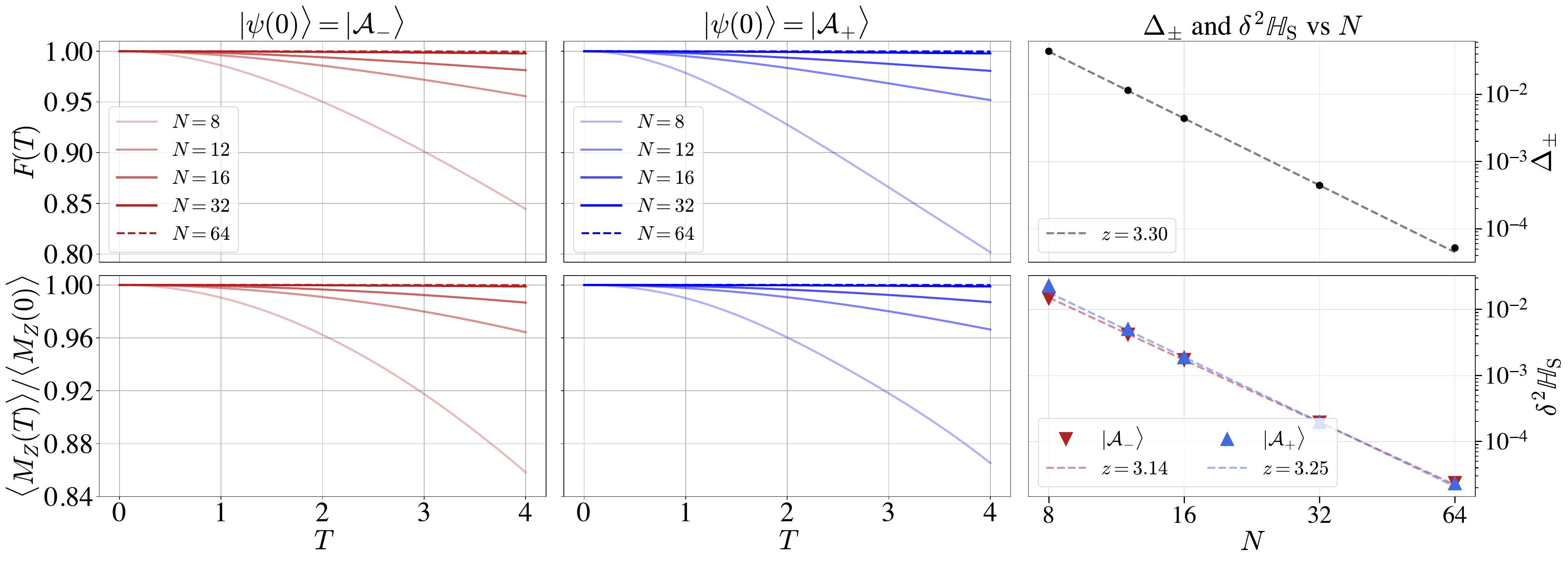}
    
    \caption{Decay of the fidelity, $F(T) = |\bra{\psi (0)} \psi (T) \rangle |^2$ (top) and the normalised total magnetisation $\langle M_Z (T) \rangle / \langle M_Z (0) \rangle$ (bottom). We show results for the two gapless excitations (of $\hat{\mathbb{H}}^{\mathrm{Fredkin}}_\mathrm{0}$) $\ket{\mathcal{A}_+}$ (left) and $\ket{\mathcal{A}_-}$ (middle). Both states display a slower decay of the fidelity and total magnetisation as we increase the system size from $N = 8$ to $N = 64$. While there are deviations for the fidelity decay at $N = 8$ between the two AQMBS states, we note that the other system sizes show great agreement, which is expected due to the same dynamical exponent for the rate of the closing of the gap in the reference Hamiltonian. On the right we display the scaling for the gap with respect to $\hat{\mathbb{H}}_0^\mathrm{Fredkin}$ (top) and the variance with respect to $\hat{\mathbb{H}}_\mathrm{S}^\mathrm{Fredkin}$ (bottom).}
    \label{fig:fredkin_results}
\end{figure*}
The Fredkin spin chain Hamiltonian \cite{Sal-16, Mov-18} is defined for a system of half-integer spins, with three local interactions. This model has the peculiar property that its ground state displays an entanglement phase transition \cite{Zha-17}, where the ground state has bounded entanglement entropy, $\mathcal{O}(1)$ on either side of the critical point, but has logarithmic scaling, $\sim \log N$, at the critical point. The gap at the latter is closing as $\sim N^{-3.2}$ \cite{Che-17}. We note that the bound on the gap is known exactly \cite{Mov-18}, but its exact scaling has only been investigated numerically. For simplicity, and ease of numerical implementation, we will consider the simplest case of sites with local Hilbert space dimension $d = 2$. The Fredkin Hamiltonian for qubits, at the critical point, can be written as:
\begin{align}
    \hat{\mathbb{H}}^{\mathrm{Fredkin}}_\mathrm{0} &= |\downarrow \rangle \langle \downarrow |_1 + |\uparrow \rangle \langle \uparrow |_N \\ &+\sum_{j=1}^{N-2}\left[\hat{P}^{\uparrow}_{j}\hat{P}^{-}_{j+1,j+2}+\hat{P}^{-}_{j,j+1}\hat{P}^{\downarrow}_{j+2}\right],
\end{align}
where the different projectors are defined as:
\begin{align}
    \hat{P}^{\uparrow}_{j}&=\frac{1+\hat{Z}_j}{2} = \ket{\uparrow} \bra{\uparrow}_j, \:  \hat{P}^{\downarrow}_{j}=\frac{1-\hat{Z}_j}{2} = \ket{\downarrow} \bra{\downarrow}_j \notag \\ \hat{P}^{-}_{i,j} &= |-\rangle\langle-|_{ij}, \quad |-\rangle_{ij} = \frac{1}{\sqrt{2}}\left(|\uparrow\downarrow\rangle -|\downarrow\uparrow\rangle \right).
\end{align}
We can also rewrite the bulk projection terms as: 
\begin{equation}
    P^{\uparrow}_{j}P^{-}_{j+1,j+2}+P^{-}_{j,j+1}P^{\downarrow}_{j+2} \equiv P_{j, j+1, j+2}^\mathrm{bulk},
\end{equation}
in order to more readily write down the corresponding scarred Hamiltonian. One can easily see that $\hat{\mathbb{H}}^{\mathrm{Fredkin}}_\mathrm{0}$ conserves the total magnetisation, $[\hat{\mathbb{H}}^{\mathrm{Fredkin}}_\mathrm{0},\hat{M}_Z] = 0, \, \hat{M}_Z \equiv \sum_n \hat{Z}_n$. The ground state of $\hat{\mathbb{H}}^{\mathrm{Fredkin}}_\mathrm{0}$ is in the $\hat{M}_Z = 0$ sector and has a known closed form expression, as a superposition of Dyck paths \cite{Sal-16}. The first excited state is doubly degenerate, with one instance occurring in each of the $\hat{M}_Z = \pm 1$ sectors. These two excitations, which we will denote as $\ket{\mathcal{A}_\pm}$ are gapless, with the gap closing as $\Delta_\pm \sim N^{-3.2}$ \cite{Che-17}. There is also a gapless excitation with $\hat{M}_Z = 0$, which we will denote as $\ket{\mathcal{A}_0}$, whose corresponding gap is closing as $\Delta_0 \sim N^{-2.7}$ \cite{Che-17}. 

As dictated by Theorem \ref{theorem:AQMBS_main_text}, we expect $\ket{\mathcal{A}_{\pm (0)}}$ to become AQMBS states of a deformed, scarred Hamiltonian, with energy variance decaying according to $\Delta_\mathrm{\pm (0)}$. For the case of the Fredkin chain, we choose the scarred Hamiltonian as:
\begin{align} \label{eq:fredkin_unc}
\hat{\mathbb{H}}^{\mathrm{Fredkin}}_\mathrm{S} &= |\downarrow \rangle \langle \downarrow |_1 + |\uparrow \rangle \langle \uparrow |_N \\ &+\sum_{j=1}^{N-2}\hat{P}_{j, j+1, j+2}^\mathrm{bulk} \hat{h}_{j,j+1,j+2} \hat{P}_{j, j+1, j+2}^\mathrm{bulk}
\end{align}
We set all of the generators in Eq.~\eqref{eq:fredkin_unc} to $\hat{h}_{j, j+1, j+2} = [\hat{1} - \mathrm{id}]_{d^3}$, where $\hat{1}$ corresponds to the matrix representation, in the computational basis, with all elements set to 1, and id is the identity. We choose this generator to capture generic behaviour, while omitting the arbitrary phase accumulation (as the term proportional to the identity would give back the reference Hamiltonian, for which $\ket{\mathcal{A}_{\pm (0)}}$ are eigenstates).

In order to find the states $\ket{\mathcal{A}_{\pm (0)}}$ numerically, we perform DMRG calculations on $\hat{\mathbb{H}}^{\mathrm{Fredkin}}_\mathrm{0}$, using the Julia ITensor MPS package \cite{itensor}. Once we have obtained the desired states, we perform time evolution under $\hat{\mathbb{H}}^{\mathrm{Fredkin}}_\mathrm{S}$ defined in Eq.~\eqref{eq:fredkin_unc}, using TEBD implemented in Julia \cite{itensor}, computing the expectation value of the total magnetisation, $\langle \hat{M}_Z (T)\rangle \equiv \langle \sum_n \hat{Z}_n \rangle$, and the fidelity with respect to the initial state, $F(T) = |\bra{\psi (0)} \psi (T) \rangle |^2$. The results are displayed in Fig. \ref{fig:fredkin_results}.

We can see that both probed states, $\ket{\mathcal{A}_\pm}$, show slower decay of the fidelity and the total magnetisation as we increase the system size through the range $N \in [8, 12, 16, 32, 64]$. It is also important to note that, while there are deviations for the fidelity decay at $N = 8$ between the two AQMBS states, the other system sizes show excellent agreement. This is as expected since, by Theorem \ref{theorem:AQMBS_main_text}, the variance decrease is dictated by the rate of the closing of the gap in the reference Hamiltonian. By the Mandelstamm-Tamm bound, introduced in Eq.~\eqref{eq:tamm_bound}, this in turn leads to a roughly similar behaviour in the fidelity decay. The results for the state $\ket{\mathcal{A}_0}$, which is in the $\hat{M}_Z = 0$ subspace, are in SM \ref{app:fredkin_mag0_aqmbs}.

Moreover, we also compute the value of the gaps $\Delta _\pm$ for the probed system sizes, as well as the variance ($\delta^2 \hat{\mathbb{H}}_{\rm S}$) with respect to the scarred Hamiltonian (displayed in last column of Fig. \ref{fig:fredkin_results}). The results show agreement in terms of scaling with respect to system size, demonstrating that the stability of AQMBS is determined by the dynamical exponent of the phase transition in $\hat{\mathbb{H}}_\mathrm{0}^\mathrm{Fredkin}$. In other words, as $\ket{\mathcal{G}}$, which is undergoing a ground state entanglement phase transition of $\hat{\mathbb{H}}_\mathrm{0}^\mathrm{Fredkin}$ is promoted to a QMBS of $\hat{\mathbb{H}}_\mathrm{S}^\mathrm{Fredkin}$, the gapless excitations of $\hat{\mathbb{H}}_\mathrm{0}^\mathrm{Fredkin}$ are promoted to AQMBS.

\subsection{Motzkin Spin Chain} \label{subsec:motzkin}

\begin{figure*}
    \centering
    \includegraphics[width=0.98\linewidth]{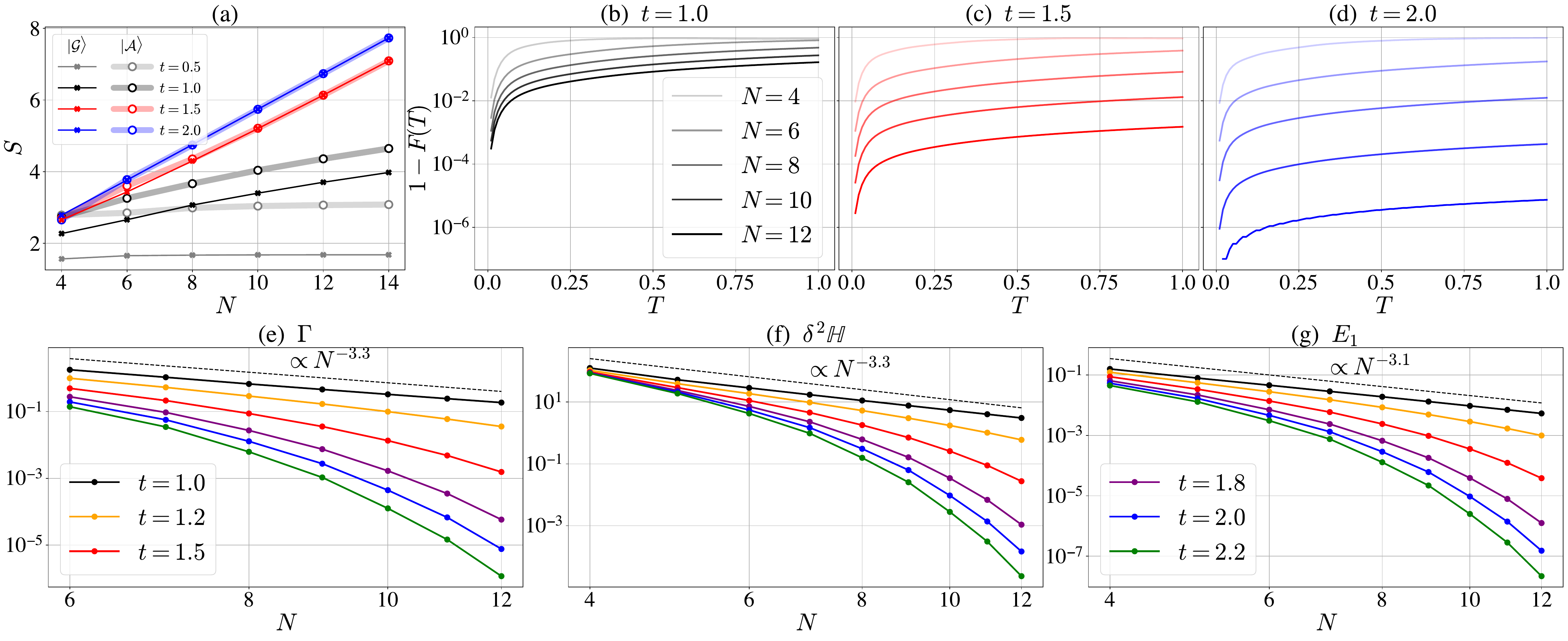}
    \caption{Results for the Motzkin model. In panel (a) we display the entanglement scaling of the ground state $\ket{\mathcal{G}}$ and the first excited state $\ket{\mathcal{A}}$ as a function of the system size $N$ for different values of the deformation parameter $t$. For $t = 0.5$, deep in the gapped regime, both states display an area law scaling of bipartite entanglement entropy, aligned with analytic results. As we increase $t$ above $t = 1$ we observe a transition to volume law scaling, which has been demonstrated for $\ket{\mathcal{G}}$ in \cite{Zha-17}, but is not known exactly for $\ket{\mathcal{A}}$. In panels (b), (c), (d) we display the infidelity increase as a function of time with the initial state $\ket{\mathcal{A}}$ for $t = 1.0, 1.5, 2.0$, respectively. We observe that at each individual $t$ value, the infidelity increase slower for increasing $N$, while the increasing $t$ values lead to drastically slower decay rates. This is corroborated by results for the fidelity decay rate $\Gamma$ as a function of $N$ (e). While $t = 1$ results in a decay rate scaling as $\sim N^{-3.3}$, larger $t$ values lead to progressively faster decays which are more rapid than any power-law description. The decay rates $\Gamma(t, N)$ are closely aligned with the scaling for variance $\delta^2 \mathbb{H}_\mathrm{S}^\mathrm{Motzkin}$, results for which are shown in (f), and consequently the scaling of the gap, $E_1$, shown in (g).}
    \label{fig:motzkin_results}
\end{figure*}

The Motzkin spin chain can be regarded as the integer-spin analogue of the Fredkin model, and was first introduced in Ref. \cite{Bra-12} for the spin-1 case. This was later generalised to an arbitrary integer spin-$s$ model, as well as the so-called weighted Motzkin spin chain parameterised by $t \in \mathbb{R}_{>0}$ \cite{Zha-17}. The Hamiltonian for the latter is:
\begin{align} \label{eq:motzkin_full}
\hat{\mathbb{H}}^\mathrm{Motzkin}_\mathrm{0}(s,t)
&=
\hat{\Pi}^{\mathrm{bd}}(s)
+\sum_{j=1}^{N-1}\hat{\Pi}^\mathrm{bulk}_{j,j+1}(s,t)
,
\\[2mm]
\hat{\Pi}^{\mathrm{bd}}(s)
&=
\sum_{c=1}^{s}
\left(
|d^c\rangle\langle d^c|_{1}
+
|u^c\rangle\langle u^c|_{N}
\right),
\\[2mm] \label{eq:same_c_terms}
\hat{\Pi}_{j,j+1}^\mathrm{bulk}(s,t)
&=
\sum_{c=1}^{s}
\Bigl(
|\Phi^c(t)\rangle\langle\Phi^c(t)|_{j,j+1}
\nonumber\\
&\qquad
+
|\Psi^c(t)\rangle\langle\Psi^c(t)|_{j,j+1}
\nonumber\\
&\qquad
+
|\Theta^c(t)\rangle\langle\Theta^c(t)|_{j,j+1}
\Bigr) \nonumber\\
&\qquad + \sum_{\substack{c,c'=1\\ c\neq c'}}^{s}
|u^c d^{c'}\rangle
\langle u^c d^{c'}|_{j,j+1},
\end{align}
where $u^c$ and $d^c$ are up and down states of colour $c$. For a spin-$s$ model, there will be $s$ colours. For example with $s=2$ we have a local Hilbert space dimension of $d = 5$ and we can write the basis as $\{ \ket{d^2}, \ket{d^1}, \ket{0}, \ket{u^1}, \ket{u^2} \}$. The parameterised states used in Eq.~\eqref{eq:same_c_terms} are \cite{Zha-17}:
\begin{align}
    |\Phi^c(t)\rangle=\frac{1}{\sqrt{1+t^2}}\left(|u^c 0\rangle-t |0 u^c\rangle\right), \\ |\Psi^c(t)\rangle=\frac{1}{\sqrt{1+t^2}}\left(|0 d^c\rangle-t |d^c 0\rangle\right),\\ |\Theta^c(t)\rangle=\frac{1}{\sqrt{1+t^2}}\left(|u^c d^c\rangle-t |00\rangle\right).
\end{align}
We note that, independently of the value of $t$, $\hat{\mathbb{H}}_\mathrm{0}^\mathrm{Motzkin}$ conserves the different colours. If we write the total ``magnetisation'' associated with colour $c$ as $\hat{M}_Z^c = \sum_n \hat{Z}_n^c$, where $\hat{Z}_n^c \equiv \ket{u^c}\bra{u^c} - \ket{d^c}\bra{d^c}$, then $[\hat{\mathbb{H}}_\mathrm{0}^\mathrm{Motzkin}, \hat{M}_Z^c] = 0$. This does not hold for $\hat{\mathbb{H}}_\mathrm{S}^\mathrm{Motzkin}$ for generic $\hat{h}_n$ however.

The ground state of $\hat{\mathbb{H}}_\mathrm{0}^\mathrm{Motzkin}$ is in the symmetry subspace where each colour magnetisation is $0$ and can be written down exactly for any values of $s \geq 1$, $t \geq 0$. It can be expressed as a weighted superposition of Motzkin paths \cite{Zha-17}, and it  undergoes an entanglement phase transition as we tune the parameter $t$. For $t < 1$, the model is gapped, with the ground state entanglement entropy obeying the area law. At $t = 1$, for $s \geq 2$ the ground state undergoes an entanglement phase transition, with its entanglement entropy scaling as $\sim \sqrt N$, and the gap closing polynomially with similar exponents as the Fredkin chain \cite{Che-17}. However, most interestingly, for $t > 1$ and $s \geq 2$, the ground state obeys a volume law scaling \cite{Zha-17}, while its gap decreases super-exponentially \cite{Lev-17}, since the gap of the first excited state $\ket{\mathcal{A}}$, within the $M_Z^c = 0 \,  \forall \, c$ subspace is upper bounded as:
\begin{equation} \label{eq:motzkin_exponential_gap}
    \Delta_0(N, t) \leq 8Nst^{-N^2/3}.
\end{equation}
Now, recalling Theorem \ref{theorem:AQMBS_main_text}, this opens up the interesting possibility of the existence of exponentially stable AQMBS in the $\hat{\mathbb{H}}^\mathrm{Motzkin}_\mathrm{S}$ set of Hamiltonians. Moreover, as the ground state exhibits volume law scaling of the entanglement entropy, it is plausible that the first excited state exhibits similar behaviour. 

In order to show the exponentially stable AQMBS and other features, we construct $\hat{\mathbb{H}}_\mathrm{S}^\mathrm{Motzkin}$ as:
\begin{equation}
\hat{\mathbb{H}}_\mathrm{S}^\mathrm{Motzkin}(s,t) = \hat{\Pi}^\mathrm{bd}(s) + \sum_{j = 1}^{N-1} \hat{\Pi}^\mathrm{bulk}_{j, j+1}(s,t) \hat{h}_{j, j+1} \hat{\Pi}^\mathrm{bulk}_{j, j+1}(s,t).
\end{equation}
For simplicity, we set $s = 2$, which is the smallest integer value for which we observe the super-exponential closing of the gap. Similarly to the Fredkin model, we set all $\hat{h}_{j, j+1}$ to be equal, and given by a random hermitian matrix without diagonal elements. The results for the obtained model are shown in Fig. \ref{fig:motzkin_results}.

First, we notice that the entanglement entropy scaling of the AQMBS state $\ket{\mathcal{A}}$ is closely aligned with the scaling for $\ket{\mathcal{G}}$. We can see that both states undergo a transition from area-law entropy scaling (at $t = 0.5$) to a volume law scaling (for $t > 1$), while displaying $\sim \sqrt{N}$ behaviour at $t = 1$. We note that, while the entropy for $\ket{\mathcal{G}}$ has been analytically computed for all values of $t$, there are no corresponding results for $\ket{\mathcal{A}}$.  

Now, as we time evolve $\ket{\mathcal{A}}$ with respect to $\hat{\mathbb{H}}_\mathrm{S}^\mathrm{Motzkin}$, we again observe improved stability with increasing $N$, demonstrated by the progressively slower increase of the infidelity, $1- F(T)$, shown in Fig. \ref{fig:motzkin_results} (b)-(d), for values $t = 1.0, 1.5, 2.0$, respectively. It is also evident that the decay rates are becoming slower for increasing $t$.

In order to quantify how the fidelity decay changes for various $t$ values, we will consider the decay rate $\Gamma$: $F(T) = \exp [-\Gamma T]$, with the results shown in Fig. \ref{fig:motzkin_results} (e). For $t = 1$, the decay rate obeys a power law scaling with $\sim N^{-3.3}$, while greater $t$ leads to progressively faster decays, showing a deviation from polynomial behaviour.  

Finally, we note that the decay rate behaviour is closely aligned with the directly computed variances $\delta ^2 \mathbb{H}_\mathrm{S}^\mathrm{Motzkin}$, shown in Fig. \ref{fig:motzkin_results} (f) and the computed gaps of $\hat{\mathbb{H}}_\mathrm{0}^\mathrm{Motzkin}$, for which the results are in Fig. \ref{fig:motzkin_results} (g). The dynamical exponent at $t = 1$ is approximately $z = 3.1$, mirroring the decay and variance, scaling as $N^{-3.3}$.

\subsection{PVBS in higher dimensions} \label{subsec:PVBS}

In Ref. \cite{Bac-12}, Bachmann and Nachtergaele introduced a class of Hamiltonians on $N$ qudit chains with open boundary conditions, called product vacua with boundary state (PVBS) models. For qubits in a one-dimensional chain the Hamiltonian can be written as: \begin{equation} \hat{\mathbb{H}}^\mathrm{1D}_\mathrm{0}(g) = \sum_{n=0}^{N-2} \hat{P}^{(g)}_{n,n+1}, \quad \hat{P}^{(g)} = \ket{11}\bra{11} + \ket{\psi^{(g)}} \bra{\psi^{(g)}} , \label{eq:1d_pvbs} \end{equation} where $\hat{P}_{n,n+1}^{(g)}$ is two-qubit projector acting non-trivially only on qubits $n$ and $n+1$, and: \begin{equation} |\psi^{(g)}\rangle = \frac{1}{\sqrt{1+|g|^2}} \big(g^*|01\rangle - |10\rangle \big) , \end{equation} is a two-qubit state parameterised by a complex variable $g$. Here, $\{ \ket{0}, \ket{1} \}$ is the single-qubit computational basis. The PVBS model can also be generalised to a higher dimensional lattice $\Lambda$:
\begin{equation} \hat{\mathbb{H}}^\mathrm{D > 1}_\mathrm{0}(g) = \sum_{\langle n, m \rangle} \hat{P}^{(g)}_{n,m} , \label{eq:higher_d_PVBS} \end{equation}
where $\langle n, m \rangle$ denotes nearest neighbour sites on $\Lambda$. For $D = 2$, considering a rectangular lattice of dimensions $L_x \times L_y$, the gapless excitations and their energies have the form (see SM \ref{app:higher_d_pvbs} for the derivation):
\begin{align} \label{eq:2d_pvbs_eigstates}
    |\mathcal{A}_{k_x, \, k_y} \rangle &= |\mathcal{A}_{k_x} \rangle \otimes |\mathcal{A}_{k_y} \rangle, \nonumber \\ |\mathcal{A}_{k_{x(y)}} \rangle &= \sum_n^{L_{x(y)}-1} \cos \left [ \left (L_{x(y)}-n-\frac{1}{2} \right) \frac{k_{x(y)}\pi}{L_{x(y)}} \right ] |n\rangle, \\ \label{eq:2d_pvbs_eigvalues} E_{k_x, \, k_y} &= 2 - \cos \left (\frac{k_x\pi}{L_x} \right ) - \cos \left (\frac{k_y\pi}{L_y} \right ),
\end{align}
which is an exact QMBS of $\hat{\mathbb{H}}_\mathrm{S}^\mathrm{D = 1}$ for $k_x = k_y = 0$, but forms an AQMBS otherwise. When constructing $\hat{\mathbb{H}}^\mathrm{D=2}_\mathrm{S}$, it is useful to note that, due to the relatively simple form of $\hat{P}^{(g = 1)}$, we can express the choice of generators in terms of the following operators acting in the $\{\ket{\psi ^{(g)}}, \, \ket{11} \}$ subspace: $\tilde{\sigma}^x = \ket{\psi^{(g)}}\bra{11} + \ket{11}\bra{\psi^{(g)}}, \tilde{\sigma}^y=i\ket{\psi^{(g)}}\bra{11} -i \ket{11}\bra{\psi^{(g)}}, \tilde{\sigma}^z= \ket{\psi^{(g)}}\bra{\psi^{(g)}} - \ket{11}\bra{11}$. We then construct the local operators as $\hat{h}_{n,m} = 4\tilde{\sigma}^x_{n, m} + 4\tilde{\sigma}^y_{n, m} - 4\tilde{\sigma}^z_{n, m}$, acting on an $L \times L$ square lattice. The results for a particular set of initial states are presented in Fig. \ref{fig:pvbs_results}.

\begin{figure*}
    \centering

    \includegraphics[width=0.98\linewidth]{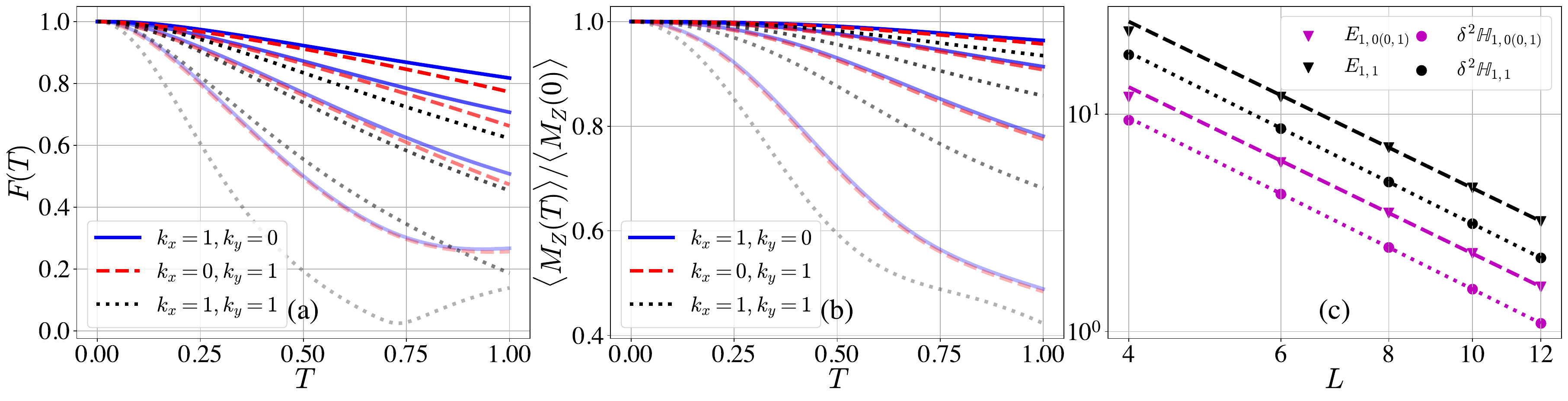}
    
    \caption{Time evolution of the fidelity $F(T)$ (a) and normalised magnetisation $M_Z / M_Z (0)$ (b) for the 2-dimensional PVBS model on an $L$ by $L$ square lattice. We probe the three AQMBS states $\ket{\mathcal{A}_{k_x, k_y}}$ with $(k_x, k_y) = (1,0), (0, 1), (1, 1)$, for $L = 4, 6, 8, 10$. The increasing system sizes are indicated by decreased opacity. Each choice of $k_x, k_y$ gives  fidelity and magnetisation relaxations that become slower with increasing $L$. While the $(k_x, k_y) = (0,1), (1,0)$ initial states are quantitatively very close for all $L$, the $(k_x, k_y) = (1,1)$ results in a faster relaxation, due to a slower closing of the gap in the reference Hamiltonian. In panel (c) we display the scaling of the energy of the probed states with respect to the reference Hamiltonian and the variance of the states with respect to the scarred Hamiltonian. }
    \label{fig:pvbs_results}
\end{figure*}

We probe the initial states $\ket{\mathcal{A}_{k_x, k_y}}$ with $(k_x, k_y) = (0,1), (1,0), (1,1)$ as defined in Eq.~\eqref{eq:2d_pvbs_eigvalues}. For all of the initial conditions, we observe diverging fidelity and relaxation timescales. We can also observe that the results for $(k_x, k_y) = (0,1), (1,0)$ are quantitatively very close, while $(k_x, k_y) = (1,1)$ always exhibits a faster decay of fidelity (magnetisation) for fixed $L$. This is aligned with the analytical predictions, as the gap with respect to the reference Hamiltonian for the latter is closing at a slower rate as can be deduced from Eq. ~\eqref{eq:2d_pvbs_eigvalues}.

The dynamics are aligned with the scaling of the variance and gaps, displayed in the rightmost panel of Fig. \ref{fig:pvbs_results}. We see that the variances with respect to the scarred Hamiltonian follow a quantitatively very close power law scaling as the gaps in the reference Hamiltonian. The former are fitted with $\sim N^{-z}, z = 1.93$ and the latter with $\sim N^{-z}, z = 1.98$. This indicates that, while we have only demonstrated a bound in Theorem \ref{theorem:AQMBS_main_text}, the scaling of the variance and the gaps appear to match for the different probed models.


\section{Discussion} \label{sec:conclusion}

To summarise, we have shown that the stability of AQMBS states in certain models can be directly related to the rate of the gap closing at critical points of phase transitions. Using this relationship allowed us to reveal AQMBS in ``scarred'' Hamiltonians of the 2-dimensional PVBS model, as well as the Fredkin and Motzkin spin chains. The latter is expecially interesting, as the AQMBS shows a transition from polynomial to exponential stability, accompanied by a change from subvolume to volume law scaling of the entropy.

While we have shown AQMBS arising from the gapless excitations of entanglement phase transitions in the Fredkin and Motzkin chains, as well as the localisation transition in 2-dimensional PVBS models, we note that the results can also be extended to topological phase transitions discussed in Ref. \cite{Sho-26}

During the preparation of the manuscript we became aware of related work \cite{Kun-26}, which proposed a connection between AQMBS and strong-to-weak spontaneous symmetry breaking, including the systematic construction of volume law AQMBS which were termed ARQMBS. While the considered constructions are different, the results indicate that using the scarred Motzkin constructions, one may be able to observe highly non-trivial hydrodynamic behaviour in corresponding open quantum systems, which we leave open for future studies.

\section{Data Availability}

The data presented in this study are available from the corresponding authors upon request.

\section{Code Availability}

The code used to obtain the results of this study are available from the corresponding authors upon request.

\begin{acknowledgments}
L.L. would like to acknowledge helpful discussions with Lukas Broers. This work is supported by JST Grant Number JPMJPF2221 (L.L., T.O., N.Y.) and IBM Quantum (L.L., N.Y.). S.D. is supported by Taighde \'{E}ireann – Research Ireland under Grant number 22/PATH-S/10812. T.O. is additionally supported by MEXT Next-Generation Computational Science Global Reach \& Innovation Program Japan (Grant Number JPMXP1920260301). N.Y. is also supported by JST CREST Grant Number JPMJCR23I4, Google Quantum AI, JST ASPIRE Grant Number JPMJAP2316, JST ERATO Grant Number JPMJER2302, JST [Moonshot R\&D] [Grant Number JPMJMS256J], and Institute of AI and Beyond of the University of Tokyo.

\end{acknowledgments}

\newpage
\clearpage

\bibliography{refs.bib}

\clearpage

\setcounter{page}{1}

\setcounter{section}{0}

\titleformat{\section}
  {\normalfont \bfseries \filcenter}
  {\thesection\ --}
  {0.5em}
  {}

\setcounter{figure}{0}
\renewcommand{\figurename}{Supplementary Figure}

\onecolumngrid 

{\large \textbf{Supplementary Material for ``Asymptotic scars arising from the remnants of quantum phase transitions''}}
{\large Leonard Logari\'c, Shane Dooley, Tsuyoshi Okubo, Nobuyuki Yoshioka}

\setcounter{equation}{0}

\section{Proof of Theorem \ref{theorem:AQMBS_main_text}}
\label{app: proof_aqmbs}

Let us consider a system of $N$ qudits, with local Hilbert space dimension $d$, labelled by $n \in \{0, 1, \dots, N-1\}$. Let $\{\hat{P}_n\}_{n=0}^{N-1}$ be a set of local projectors, where $\hat{P}_n$ is supported on a set of sites $X_n$ containing site $n$ and at most $l$ other sites, $|X_n| \le l+1$. We denote by
\begin{equation} \label{eq:z_def}
    z = \max_n \bigl| \{ m \neq n : X_m \cap X_n \neq \varnothing \} \bigr|
\end{equation}
the maximal number of projectors whose support overlaps with that of a given projector, and assume that $z$ is independent of $N$. For example, in one dimension with contiguous supports $X_n = \{n, n+1, \dots, n+l\}$, one has $z = 2l$.

Using these projectors, we define the $(l+1)$-local reference Hamiltonian
\begin{equation} \label{eq:H0}
    \hat{\mathbb{H}}_0 = \sum_n \hat{P}_n .
\end{equation}
Since the projectors are positive semi-definite, so is $\hat{\mathbb{H}}_0$. Using a deformation, we transform the reference Hamiltonian into a ``scarred'' Hamiltonian
\begin{equation} \label{eq:HS}
    \hat{\mathbb{H}}_\mathrm{S} = \sum_n \hat{P}_n \hat{h}_n \hat{P}_n = \sum_n \hat{H}_n ,
\end{equation}
where $\hat{h}_n$ is an arbitrary Hermitian operator whose support is contained within $X_n$, with bounded spectral norm, $\|\hat{h}_n\| \le B$. We impose that the reference Hamiltonian is frustration-free, i.e., its ground state(s) $\ket{\mathcal{G}}$ are simultaneously annihilated by every $\hat{P}_n$, so that the ground state energy is $E_\mathcal{G} = 0$. We assume that the sets of operators $ \{ \hat{P}_n \}$, $ \{ \hat{h}_n \}$ are dimensionless, such that both $\hat{\mathbb{H}}_0$ and $\hat{\mathbb{H}}_\mathrm{S}$ are dimensionless as well. The proof is easily adapted for the dimensionful case, as we remark at the end of this section.

We assume that $\hat{\mathbb{H}}_0$ has gapless excitations, $\ket{\mathcal{A}_k}$, which are normalised eigenstates whose eigenvalues $\Delta_k(N)$, measured relative to $E_\mathcal{G} = 0$, are monotonically decreasing with $N$ and vanish in the thermodynamic limit, $\lim_{N \to \infty} \Delta_k(N) = 0$. This allows us to formulate the main theorem of our work.

\begin{theorem} \label{thm:main}
    Let $\hat{\mathbb{H}}_0$ be a frustration-free, $(l+1)$-local Hamiltonian with eigenstates $\ket{\mathcal{A}_k}$ whose eigenvalues $\Delta_k(N)$  vanish in the thermodynamic limit. Then the states $\ket{\mathcal{A}_k}$ are AQMBS of $\hat{\mathbb{H}}_\mathrm{S} = \sum_n \hat{P}_n \hat{h}_n \hat{P}_n$. Specifically, their energy and energy variance with respect to $\hat{\mathbb{H}}_\mathrm{S}$ obey
    \begin{equation} \label{eq:main_bounds}
        \bigl| \langle \hat{\mathbb{H}}_\mathrm{S} \rangle \bigr| \le B \, \Delta_k(N), \qquad
        \langle \hat{\mathbb{H}}_\mathrm{S}^2 \rangle - \langle \hat{\mathbb{H}}_\mathrm{S} \rangle^2 \le B^2 \bigl[ \Delta_k^2(N) + 2z \, \Delta_k(N) \bigr],
    \end{equation}
    where $\langle \cdot \rangle \equiv \bra{\mathcal{A}_k} \cdot \ket{\mathcal{A}_k}$, $B$ bounds the spectral norm of $\hat{h}_n$, and $z$ is defined in Eq.~\eqref{eq:z_def}. Both quantities therefore vanish in the thermodynamic limit, as $\Delta_k(N) \to 0$.
\end{theorem}

\begin{proof}
We first establish three operator inequalities, where $\hat{A} \le \hat{C}$ means that $\hat{C} - \hat{A}$ is positive semi-definite.

\emph{(a) Bounds on local operators $\hat{H}_n$.} Conjugating $-B \le \hat{h}_n \le B$ by $\hat{P}_n$ gives
\begin{equation} \label{eq:local_h_bound}
    -B \hat{P}_n \le \hat{H}_n \le B \hat{P}_n .
\end{equation}
Moreover, $\hat{P}_n \le 1$ implies $\hat{h}_n \hat{P}_n \hat{h}_n \le \hat{h}_n^2 \le B^2$, and therefore
\begin{equation} \label{eq:local_h_sq_bound}
    \hat{H}_n^2 = \hat{P}_n \hat{h}_n \hat{P}_n \hat{h}_n \hat{P}_n \le \hat{P}_n \hat{h}_n^2 \hat{P}_n \le B^2 \hat{P}_n .
\end{equation}

\emph{(b) Anticommutators.} For any Hermitian $\hat{A}$ and $\hat{C}$, positivity of $(\hat{A} \pm \hat{C})^2$ gives $\hat{A}^2 + \hat{C}^2 \pm \{\hat{A}, \hat{C}\} \ge 0$, i.e.,
\begin{equation} \label{eq:anticomm_bound}
    -(\hat{A}^2 + \hat{C}^2) \le \{\hat{A}, \hat{C}\} \le \hat{A}^2 + \hat{C}^2 .
\end{equation}

\emph{(c) $\hat{H}_n$, $\hat{H}_m$ with disjoint supports.} If $X_n \cap X_m = \varnothing$, then $\hat{H}_n$ and $\hat{H}_m$ commute, so $\{\hat{H}_n, \hat{H}_m\} = 2 \hat{H}_n \hat{H}_m$ is Hermitian. Furthermore, $\hat{P}_n \hat{P}_m$ is a projector, $\hat{H}_n \hat{H}_m = \hat{P}_n \hat{P}_m \, \hat{H}_n \hat{H}_m \, \hat{P}_n \hat{P}_m$, and $\|\hat{H}_n \hat{H}_m\| \le \|\hat{h}_n\| \, \|\hat{h}_m\| \le B^2$. Hence, for any state $\ket{\psi}$, $\bra{\psi} \hat{H}_n \hat{H}_m \ket{\psi} \le B^2 \| \hat{P}_n \hat{P}_m \ket{\psi} \|^2 = B^2 \bra{\psi} \hat{P}_n \hat{P}_m \ket{\psi}$, i.e.,
\begin{equation} \label{eq:disjoint_bound}
    \{\hat{H}_n, \hat{H}_m\} \le 2 B^2 \hat{P}_n \hat{P}_m = B^2 \{\hat{P}_n, \hat{P}_m\} .
\end{equation}

\emph{Bound on energy.} Since $\ket{\mathcal{A}_k}$ is a normalised eigenstate of $\hat{\mathbb{H}}_0$, we have $\sum_n \langle \hat{P}_n \rangle = \langle \hat{\mathbb{H}}_0 \rangle = \Delta_k(N)$. Using Eq.~\eqref{eq:local_h_bound},
\begin{equation}
    \bigl| \langle \hat{\mathbb{H}}_\mathrm{S} \rangle \bigr| \le \sum_n \bigl| \langle \hat{H}_n \rangle \bigr| \le B \sum_n \langle \hat{P}_n \rangle = B \, \Delta_k(N) .
\end{equation}

\emph{Bound on variance.} Since $\langle \hat{\mathbb{H}}_\mathrm{S} \rangle^2 \ge 0$, the variance is bounded by $\langle \hat{\mathbb{H}}_\mathrm{S}^2 \rangle = \| \hat{\mathbb{H}}_\mathrm{S} \ket{\mathcal{A}_k} \|^2$. We split this into diagonal, non-overlapping, and overlapping contributions:
\begin{equation} \label{eq:h_expansion}
    \langle \hat{\mathbb{H}}_\mathrm{S}^2 \rangle
    = \underbrace{\sum_n \langle \hat{H}_n^2 \rangle}_{C_1}
    + \underbrace{\sum_{\substack{n < m \\ X_n \cap X_m = \varnothing}} \langle \{\hat{H}_n, \hat{H}_m\} \rangle}_{C_2}
    + \underbrace{\sum_{\substack{n < m \\ X_n \cap X_m \neq \varnothing}} \langle \{\hat{H}_n, \hat{H}_m\} \rangle}_{C_3} .
\end{equation}

For the diagonal term, Eq.~\eqref{eq:local_h_sq_bound} immediately gives:
\begin{equation} \label{eq:c1_bound}
    C_1 \le B^2 \sum_n \langle \hat{P}_n \rangle = B^2 \Delta_k(N) .
\end{equation}

For the overlapping term, Eqs.~\eqref{eq:local_h_sq_bound} and ~\eqref{eq:anticomm_bound} give:
\begin{equation} \label{eq:c3_bound}
    C_3 \le \sum_{\substack{n < m \\ X_n \cap X_m \neq \varnothing}} \langle \hat{H}_n^2 + \hat{H}_m^2 \rangle
    \le B^2 \sum_{\substack{n < m \\ X_n \cap X_m \neq \varnothing}} \langle \hat{P}_n + \hat{P}_m \rangle
    = B^2 \sum_n z_n \langle \hat{P}_n \rangle
    \le z B^2 \Delta_k(N),
\end{equation}
where $z_n = | \{ m \neq n : X_m \cap X_n \neq \varnothing \} | \le z$. The equality holds because each overlapping pair contributes both $\langle \hat{P}_n \rangle$ and $\langle \hat{P}_m \rangle$, so $\langle \hat{P}_n \rangle$ appears once for every $m \neq n$ overlapping with it, whether $m > n$ or $m < n$. 

For the non-overlapping term, Eq.~\eqref{eq:disjoint_bound} gives
\begin{equation} \label{eq:c2_bound}
    C_2 \le B^2 \mathcal{J}_N, \qquad
    \mathcal{J}_N \equiv \sum_{\substack{n < m \\ X_n \cap X_m = \varnothing}} \langle \{\hat{P}_n, \hat{P}_m\} \rangle .
\end{equation}
To bound $\mathcal{J}_N$, we expand $\langle \hat{\mathbb{H}}_0^2 \rangle$ in the same way. Using the fact that $\ket{\mathcal{A}_k}$ is a normalised eigenstate, $\langle \hat{\mathbb{H}}_0^2 \rangle = \Delta_k^2(N)$, and that $\hat{P}_n^2 = \hat{P}_n$, the diagonal terms sum to $\sum_n \langle \hat{P}_n \rangle = \Delta_k(N)$. Hence:
\begin{equation} \label{eq:reference_h_sq_expansion}
    \Delta_k^2(N) = \Delta_k(N) + \mathcal{J}_N + \mathcal{S}_N, \qquad
    \mathcal{S}_N \equiv \sum_{\substack{n < m \\ X_n \cap X_m \neq \varnothing}} \langle \{\hat{P}_n, \hat{P}_m\} \rangle .
\end{equation}
Applying Eq.~\eqref{eq:anticomm_bound} and $\hat{P}_n^2 = \hat{P}_n$, and counting pairs as in Eq.~\eqref{eq:c3_bound},
\begin{equation} \label{eq:sn_bound}
    |\mathcal{S}_N| \le \sum_{\substack{n < m \\ X_n \cap X_m \neq \varnothing}} \langle \hat{P}_n + \hat{P}_m \rangle = \sum_n z_n \langle \hat{P}_n \rangle \le z \, \Delta_k(N) .
\end{equation}
Rearranging Eq.~\eqref{eq:reference_h_sq_expansion} then yields
\begin{equation} \label{eq:jn_bound}
    \mathcal{J}_N = \Delta_k^2(N) - \Delta_k(N) - \mathcal{S}_N \le \Delta_k^2(N) + (z - 1) \Delta_k(N) .
\end{equation}

Combining Eqs.~\eqref{eq:c1_bound}, \eqref{eq:c3_bound}, \eqref{eq:c2_bound} and~\eqref{eq:jn_bound}, we obtain
\begin{equation}
    \langle \hat{\mathbb{H}}_\mathrm{S}^2 \rangle - \langle \hat{\mathbb{H}}_\mathrm{S} \rangle^2
    \le \langle \hat{\mathbb{H}}_\mathrm{S}^2 \rangle
    \le B^2 \Delta_k(N) + B^2 \bigl[ \Delta_k^2(N) + (z-1) \Delta_k(N) \bigr] + z B^2 \Delta_k(N)
    = B^2 \bigl[ \Delta_k^2(N) + 2z \, \Delta_k(N) \bigr],
\end{equation}
as claimed. For $\Delta_k(N) \ll 2z$, the bound reduces to leading order to $2z B^2 \Delta_k(N)$; in one dimension with contiguous supports ($z = 2l$), the full bound reads $B^2 [\Delta_k^2(N) + 4l \, \Delta_k(N)]$.
\end{proof}

\noindent \emph{Remark (energy units).} If the reference Hamiltonian is instead written with an explicit energy scale, $\hat{\mathbb{H}}_0 = J \sum_n \hat{P}_n$, its eigenvalues $\Delta_k(N)$ carry units of energy, and Theorem~\ref{thm:main} holds with $\Delta_k(N)$ replaced by the dimensionless ratio $\Delta_k(N)/ J$ throughout, e.g., $\langle \hat{\mathbb{H}}_\mathrm{S}^2 \rangle - \langle \hat{\mathbb{H}}_\mathrm{S} \rangle^2 \le B^2 [ (\Delta_k/ J)^2 + 2z \, \Delta_k/ J ]$. The convention above corresponds to $J = 1$.

\section{Derivation of AQMBS for higher dimensional PVBS models} \label{app:higher_d_pvbs}

We recall that the PVBS Hamiltonian in one dimension is constructed from the local projectors $\hat{P}(g)$ as:
\begin{equation}
    \hat{P}^{(g)}= |11\rangle\langle 11 | + |\psi^{(g)}\rangle \langle \psi^{(g)}|, \quad |\psi^{(g)}\rangle = \frac{1}{\sqrt{1+|g|^2}}(g^*|01\rangle - |10\rangle), \quad \hat{\mathbb{H}}_\mathrm{0}^{1\mathrm{D}}(g) = \sum_{n = 0}^{N-2} \hat{P}_{n, n+1}^{(g)},
\end{equation}
where $\hat{P}^{(g)}_{n, n+1}$ are acting on nearest neighbour sites of a chain with sites labelled $n \in \{0, 1, 2, ..., N-1 \}$. At $ |g| = 1$, i.e. $ g = e^{i\varphi}$, this model hosts gapless excitations of the form \cite{Log-26}:
\begin{equation}
     |\mathcal{A}_{k} \rangle = \frac{1}{\mathcal{N}_k}\sum_{n = 0}^{N-1} e^{in\varphi} \cos[(L-n-\frac{1}{2}) \frac{k\pi}{N}] |n\rangle
\quad E_{k} = 1 - \cos(\frac{k\pi}{N}),
\end{equation}
where $\mathcal{N}_k$ is a normalisation constant and $\ket{n} \equiv \hat{X}_n \ket{0}^{\otimes N}$ is a basis state in the single excitation subspace.

The PVBS model can also be extended to higher, $D$ dimensions \cite{Bac-15}, with independent values of the parameter $g_{x_i}$ along each dimension $\{x_i \}_{i =1}^D$. For simplicity, since we are interested in the gapless point $|g_{x_i}| = 1$ on a regular lattice, we will set all of them to the same value $g$ for the purposes of deriving the form of AQMBS states, although it can be readily extended to the general gapless point $g_{x_i} = e^{i\varphi_{x_i}}$. We can then write the Hamiltonian as:
\begin{equation}
    \hat{\mathbb{H}}_\mathrm{0}^{\mathrm{D} > 1}(g) = \sum_{\langle n, m \rangle} \hat{P}^{(g)}_{n,m},
\end{equation}
where $\langle n, m \rangle$ denotes pairs of nearest neighbour sites. Since the projectors preserve the total magnetisation, $[\hat{P}(g), \sum_n \hat{Z}_n] = 0$, we can restrict ourselves to the single excitation subspace. For $\mathrm{D} = 2$, on a rectangular lattice of size $L_x \times L_y$ at $g = 1$, the Hamiltonian in the single excitation subspace can be reduced to:
\begin{align} \label{eq:pvbs_2d_ses_app}
    \hat{\mathbb{H}}^{\mathrm{D} = 2}_{\mathrm{0, SES}} = \sum_{n=0}^{L_x-1}\sum_{m=0}^{L_y-1} (w_n^x + w_m^y) \,|n,m\rangle\langle n,m|- \sum_{n=0}^{L_x-2}\sum_{m=0}^{L_y-1} \frac{1}{2}\left(|n,m\rangle\langle n+1,m| + |n+1,m\rangle\langle n,m|\right)\notag \\- \sum_{n=0}^{L_x-1}\sum_{m=0}^{L_y-2} \frac{1}{2}\left(|n,m\rangle\langle n,m+1| + |n,m+1\rangle\langle n,m|\right),
\end{align}
where we have introduced:
\begin{equation}
    w_n^x = \begin{cases} 0.5 & n = 0 \text{ or } n = L_x-1 \\ 1 & \text{otherwise} \end{cases}, \qquad w_m^y = \begin{cases} 0.5 & m = 0 \text{ or } m = L_y-1 \\ 1 & \text{otherwise} \end{cases}.
\end{equation}
The first term in Eq.~\eqref{eq:pvbs_2d_ses_app} corresponds to on-site energies, while the latter two terms represent hopping between sites along the $x$ and $y$ dimensions, respectively. The eigenstates and corresponding eigenvalues will have the form:
\begin{equation}
    |\mathcal{A}_{k_x, \, k_y} \rangle = |\mathcal{A}_{k_x} \rangle \otimes |\mathcal{A}_{k_y} \rangle, \quad |\mathcal{A}_{k_{x(y)}} \rangle = \sum_n^{L_{x(y)}-1} \cos[(L_{x(y)}-n-\frac{1}{2}) \frac{k_{x(y)}\pi}{L_{x(y)}}] |n\rangle, \quad E_{k_x, \, k_y} = 2 - \cos(\frac{k_x\pi}{L_x}) - \cos(\frac{k_y\pi}{L_y}).
\end{equation}
We note that for $k_x = k_y = 0$, we obtain a ground state of the model, while any other combination of $\{k_x, k_y \}$ will give a gapless excitation $|\mathcal{A}_{k_x, \, k_y} \rangle$. These will form AQMBS states of the deformed, scarred Hamiltonian $\hat{\mathbb{H}}_{\mathrm{S}}^{\mathrm{D} = 2}$. We can perform the same derivation for $\mathrm{D} = 3$ and higher dimensions. For $\mathrm{D} = 3$, the eigenstates and corresponding eigenvalues will have the form:
\begin{align}
    |\mathcal{A}_{\vec{k}} \rangle = |\mathcal{A}_{k_x} \rangle \otimes |\mathcal{A}_{k_y} \rangle \otimes |\mathcal{A}_{k_z} \rangle, \quad |\mathcal{A}_{k_{j}} \rangle &= \sum_n^{L_{j}-1} \cos[(L_{j}-n-\frac{1}{2}) \frac{k_{j}\pi}{L_{j}}] |n\rangle, j \in \{x, y, z \}, \\ E_{\vec{k}} &= 3 - \cos(\frac{k_x\pi}{L_x}) - \cos(\frac{k_y\pi}{L_y}) - \cos(\frac{k_z\pi}{L_z}),
\end{align}
where we used the shorthand $\vec{k} = (k_x, k_y, k_z)$. For $D = 2, 3$ we can thus write down the eigenstates and  eigenvalues as:
\begin{align}
    |\mathcal{A}_{\vec{k}} \rangle =  \bigotimes_{j = 1}^D|\mathcal{A}_{k_j} \rangle, \quad |\mathcal{A}_{k_{j}} \rangle &= \sum_n^{L_{j}-1} \cos[(L_{j}-n-\frac{1}{2}) \frac{k_{j}\pi}{L_{j}}] |n\rangle, 1 \leq j \leq D, \\ E_{\vec{k}} &= D - \sum_{j = 1}^D \cos(\frac{k_j\pi}{L_j}),
\end{align}

\section{Results for AQMBS in zero magnetisation sector of the Fredkin chain} \label{app:fredkin_mag0_aqmbs}

\begin{figure}[h!]
    \centering
    \includegraphics[width=0.98\linewidth]{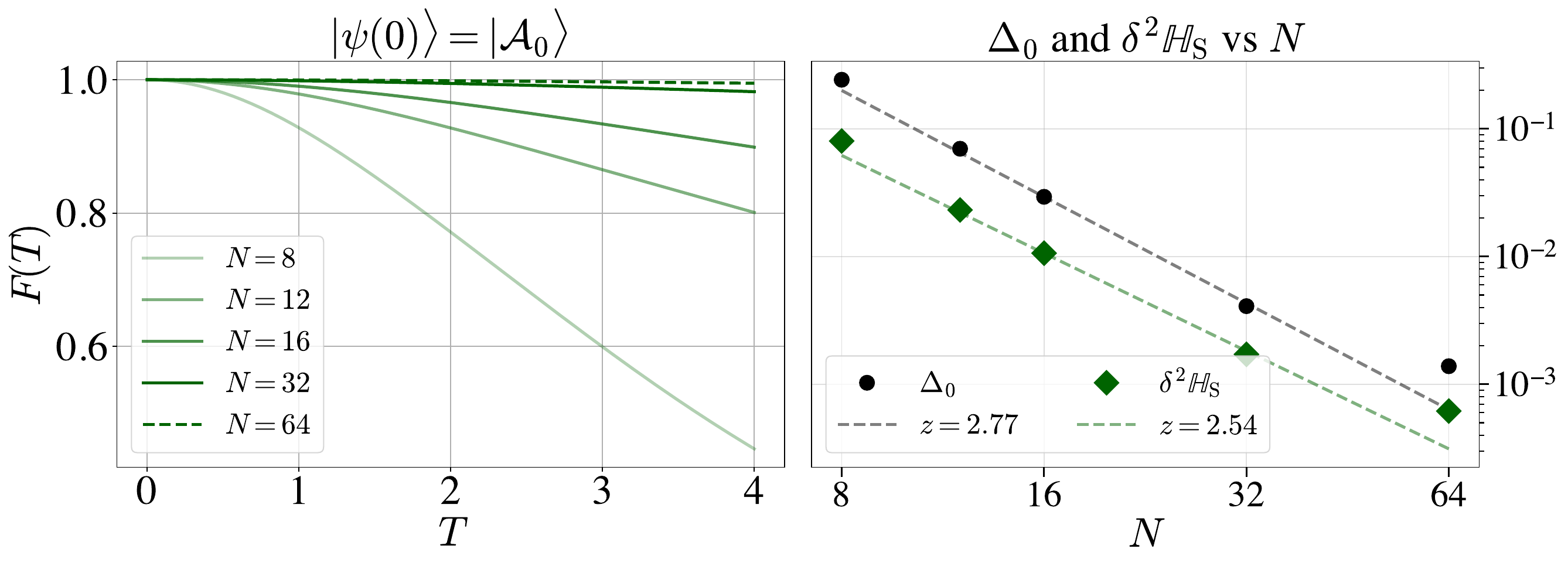}
    \caption{Left: Fidelity decay of the AQMBS in the zero magnetisation sector of the Fredkin chain, displaying a diverging relaxation time as system size is increased. Right: Scaling of the gap, $\Delta_0$ of $\ket{\mathcal{A}_0}$ with respect to $\hat{\mathbb{H}}_0^\mathrm{Fredkin}$ and the variance of $\ket{\mathcal{A}_0}$ with respect to $\hat{\mathbb{H}}_\mathrm{S}^\mathrm{Fredkin}$}
    \label{fig:fredkin_mag0_results}
\end{figure}

\end{document}

%% file: preamble.tex
\usepackage{amsmath}
\usepackage{amsthm}
\usepackage{thmtools}
\usepackage{amsfonts}
\usepackage{amssymb}
\usepackage{graphicx}
\usepackage{color}
\usepackage{accents}
\usepackage[colorlinks=true, citecolor=cyan]{hyperref}
\usepackage{enumitem}
\usepackage{titlesec}
\usepackage{tikz}
\usepackage{orcidlink}

\newcommand{\ket}[1]{|#1\rangle}
\newcommand{\bra}[1]{\langle #1|}

\usepackage{tikz}

\makeatletter 
\newsavebox{\@brx}
\newcommand{\llangle}[1][]{\savebox{\@brx}{\(\m@th{#1\langle}\)}%
  \mathopen{\copy\@brx\kern-0.5\wd\@brx\usebox{\@brx}}}
\newcommand{\rrangle}[1][]{\savebox{\@brx}{\(\m@th{#1\rangle}\)}%
  \mathclose{\copy\@brx\kern-0.5\wd\@brx\usebox{\@brx}}}
\makeatother

\newlength{\dhatheight} 

\newtheoremstyle{romanbody}
  {\topsep}       
  {\topsep}       
  {\normalfont}   
  {}              
  {\itshape}      
  {.}             
  {0.5em}         
  {}              

\theoremstyle{romanbody}

\newtheorem{theorem}{Theorem}

\renewenvironment{proof}[1][Proof.]{\begin{trivlist}
\item[\hskip \labelsep {\bfseries #1}]}{\end{trivlist}}

\renewcommand{\qed}{\nobreak \ifvmode \relax \else
      \ifdim\lastskip<1.5em \hskip-\lastskip
      \hskip1.5em plus0em minus0.5em \fi \nobreak
      \vrule height0.75em width0.5em depth0.25em\fi}